\documentclass[12pt,a4paper,final]{iopart}
\usepackage{iopams}

\usepackage{amsthm}
\usepackage{yfonts}

\usepackage{bbm}
\usepackage{bm}
\usepackage{mathrsfs}

\usepackage{enumerate}

\usepackage[normalem]{ulem}
\usepackage{color}
\definecolor{blue}{rgb}{0,0,1}
\definecolor{red}{rgb}{1,0,0}

\newcommand{\I}{\mathrm{i}}
\newcommand{\N}{\mathbb{N}}
\newcommand{\RR}{\mathbb{R}}
\newcommand{\M}{\mathcal{M}}
\newcommand{\K}{\mathcal{K}}
\newcommand{\ELL}{\mathcal{L}}
\newcommand{\EE}{\mathcal{E}}
\newcommand{\EEn}{\mathcal{E}_{n}}
\newcommand{\EEx}{\mathcal{E}^{\times}}
\newcommand{\EExn}{\mathcal{E}^{\times}_{n}}
\newcommand{\EEvn}{\mathcal{E}^{\vee}_{n}}
\newcommand{\DD}{\mathcal{D}}
\newcommand{\DDz}{\mathcal{D}{\setminus}\{0\}}
\newcommand{\DDn}{\mathcal{D}_{n}}
\newcommand{\DDx}{\mathcal{D}^{\times}}
\newcommand{\DDxn}{\mathcal{D}^{\times}_{n}}
\newcommand{\DDvn}{\mathcal{D}^{\vee}_{n}}
\newcommand{\Wvn}{W^{\vee}_{n}}
\newcommand{\CC}{\mathscr{C}}
\newcommand{\Sym}{\mathop{\mathrm{Sym}}\nolimits}
\newcommand{\Piv}{\Pi^{\vee}}
\newcommand{\Ran}{\mathop{\mathrm{Ran}}}
\newcommand{\Map}{\mathop{\mathrm{Map}}}
\newcommand{\Span}{\mathop{\mathrm{Span}}}

\DeclareMathSymbol{\shortminus}{\mathbin}{AMSa}{"39}

\newcommand{\defin}[1]{\textit{\boldmath\textbf{#1}}}

\makeatletter
\long\def\@makefntext#1{\parindent 1em\noindent 
 \makebox[1em][l]{\footnotesize\rm$\m@th{}^{\arabic{footnote}}$}%
 \footnotesize\rm #1}
\def\@makefnmark{\hbox{$
{}^{\arabic{footnote}}\m@th$}}
\def\@thefnmark{\arabic{footnote}}
\makeatother

\theoremstyle{plain}
\newtheorem{theorem}{Theorem}
\newtheorem{lemma}[theorem]{Lemma}

\newtheorem{corollary}[theorem]{Corollary}
\theoremstyle{thmstylethree}
\newtheorem{definition}[theorem]{Definition}
\newtheorem{remark}[theorem]{Remark}

\begin{document}

\title[On the running and the UV limit of Wilsonian renormalization group flows]{On the running and the UV limit of Wilsonian renormalization group flows}

\author{Andr\'as L\'aszl\'o}
\address{
Dep.\ of High Energy Physics, HUN-REN Wigner Research Centre for Physics\\
Konkoly-Thege M u 29-33, 1121 Budapest, Hungary
}
\ead{laszlo.andras@wigner.hun-ren.hu}

\author{Zsigmond Tarcsay}
\address{
Dep.\ of Mathematics, Corvinus University of Budapest\\
F\H ov\'am t\'er 13-15, 1093 Budapest, Hungary\\
and Dep.\ of Appl.\ Analysis and Comp.\ Mathematics, E\"otv\"os University\\
P\'azm\'any P\'eter s\'et\'any 1/C, 1117 Budapest, Hungary
}
\ead{zsigmond.tarcsay@uni-corvinus.hu}

\begin{abstract}
In nonperturbative formulation of quantum field theory (QFT), 
the vacuum state is characterized by the Wilsonian renormalization group (RG) 
flow of Feynman type field correlators. 
Such a flow is a parametric family of ultraviolet (UV) regularized field correlators, 
the parameter being the strength of the UV regularization, and the instances with different strength of UV regularizations 
are linked by the renormalization group equation (RGE). 
Important RG flows are those which reach out to any UV regularization strengths. 
In this paper it is shown that for these flows a natural, mathematically rigorous generally covariant 
definition can be given, and that they form a topological vector space which is 
Hausdorff, locally convex, complete, nuclear, semi-Montel, Schwartz. 
That is, they form a generalized function space having favorable properties, similar to multivariate distributions. 
The other theorem proved in the paper is that for Wilsonian RG flows 
reaching out to all UV regularization strengths, a simple factorization formula holds 
in case of bosonic fields over flat (affine) spacetime: the flow always originates from a 
regularization-independent distributional correlator, and its running 
satisfies an algebraic ansatz. 
The conjecture is that this factorization theorem should generically hold, 
which is worth future investigations.
\end{abstract}


\maketitle

\setcounter{footnote}{0}

\section{Introduction}
\label{secIntro}

The mathematically sound formulation of interacting quantum field theory 
(QFT) is a long pursued subject \cite{Todorov1990,Henneaux1994,Fewster2020}. 
Despite the difficulties encountered 
with the mathematization of the generic theory over continuum spacetimes, 
several gradual successes were reached in the past decades with the perturbative approach. 
A subfield of constructive mathematical QFT, called 
perturbative algebraic QFT (pAQFT) emerged during the past decades 
\cite{Hollands2002,Dutsch2001,Brunetti2003,Dabrowski2014,Dutsch2019}. 
In that framework the key mathematical problematics 
is the handling of spacetime pointwise products of distributional fields (propagators). 
Using advanced distribution theory, it was understood that two 
mathematical techniques are key for that. 
The so-called H\"ormander wave front set criterion is used as a sufficient 
condition on the multiplicability of distributions. Whenever that is not 
enough, an extendability theorem of distributions to singularity points is used, 
given that they have appropriate behavior (finite scaling degree) 
against spacetime stretching around those points by some control scale. 
This latter technique is relied upon, when mathematizing the usual perturbative renormalization 
of informal QFT: the coupling constants of the classical model 
are replaced by functionals of a length (or frequency) control scale, and the 
most stubborn divergences of the theory are absorbed via the running of the couplings. 
Thus, an avenue opened for formalizing the notion of perturbative 
renormalization group (pRG). 
An important milestone was the proof of perturbative renormalizability of 
Yang--Mills interactions over fixed globally hyperbolic spacetimes \cite{Hollands2008}. 
A generally covariant pAQFT framework along with a corresponding pRG formalism 
was developed partly motivated by that \cite{Costello2011,Angelo2024}.

On the rigorous \emph{nonperturbative} formulation of QFT, much less is known. 
The consensus is that for a constructive approach, the Feynman functional 
integral formulation, or a rigorous analogy of that, is needed 
\cite{Glimm1987,Velhinho2017,Albeverio2008,Gill2008,Montaldi2017}. 
That approach aims to synthetize the (possibly non-unique) vacuum 
correlators of a QFT model as the moments (or formal moments) of the 
Feynman measure (or a rigorous analogy of that), derived from a classical 
action. For interacting models, however, that approach again runs into the issue of divergences 
caused by the problematics of the multiplication of distributional fields. 
Wilson and contemporaries addressed this by the Wilsonian regularization, 
i.e.\ considering Feynman functional 
integral on a smaller subspace, namely on ultraviolet (UV) damped fields. 
Since such a subspace is obtained via coarse-graining, i.e.\ local averaging 
of fields, physicswise it is natural to require instances with 
subsequent coarse-grainings to be compatible with each-other, 
thus the notion of Wilsonian renormalization group (RG) emerged 
\cite{Wilson1975,Wilson1983,Peskin1995,Polonyi2003,Polonyi2008,Srednicki2007,Skinner2018,Dupuis2021,Kopietz2010,Bauerschmidt2019}. 
A Feynman measure instance with a given UV regularization 
is linked to a stronger UV regularized instance by ``integrating~out'' 
high frequency modes in between, 
called to be the Wilsonian renormalization group equation (RGE). 
``Integrating~out'' high frequency modes means taking the pushforward measure 
by a field coarse-graining operator, or in probability theory speak, taking the marginal measure along that. 
As it is well known, the definition of a genuine Feynman measure is problematic 
in Lorentz signature, and especially in a generally covariant setting 
\cite{Hollands2008,Feldbrugge2017a,Feldbrugge2017b,Baldazzi2019,Feldbrugge2023}. 
In order to mitigate this issue, the Feynman measure formulation and the 
corresponding RGE is usually translated to the language of formal moments, 
i.e.\ to the collection of Feynman type $n$-field correlators ($n=0,1,2,\dots$). 
That description is meaningful in arbitrary signature and also in a generally covariant setting. 
In the present paper, we prove structural theorems regarding the space of 
Wilsonian RG flows of Feynman correlators.\footnote{As 
in the usual Feynman functional integral formulation, any reference to non-Feynman 
type field correlators or to other eventually present external data, 
such as a fixed background Lorentzian causal structure, time ordering, etc, 
will be deliberately avoided.}

First, we recall the mathematical reason why in the case of 
interacting theories one is forced to define the Wilsonian regularized 
Feynman measure instead of just a Feynman measure, even in an Euclidean signature setting 
(see a concise review in \cite{Velhinho2017}). 
Take an Euclidean classical field theory, and assume that its action 
functional can be split as $S=T+V$, with $T$ being a quadratic positive definite kinetic 
term and $V$ being a higher than quadratic degree interaction term bounded from 
below.\footnote{For instance, a typical kinetic term can be an 
Euclidean Klein--Gordon term $T(\varphi)=\int \varphi\,(-\Delta+m^{2})\varphi$, 
whereas a typical interaction term can be like $V(\varphi)=g\int\varphi^{4}$.} 
Assume moreover, that the underlying spacetime manifold is an 
affine space (i.e., $\sim\RR^{N}$) so that Schwartz's functions and tempered distributions are defined, 
or alternatively, assume that the base manifold is compact 
(with regular enough boundary). 
Then, by means of Bochner--Khinchin theorem, the kinetic 
term $T$ induces a Gaussian measure $\gamma_{T}$ on the space of (tempered) 
distributional fields, see e.g.\ \cite{Velhinho2017}~Corollary~1 and its explanation on this well-known result. 
This Gaussian measure is a proper non-negative valued finite Borel measure 
under the above assumptions, devoid of any issues, and it is the Feynman 
measure of the non-interacting model. 
It is customary to write $\int(\dots)\,\mathrm{d}\gamma_{T}(\phi)$ informally as 
$\int(\dots)\,\e^{-T(\phi)}\,\mathrm{d}\phi$, as if a Lebesgue (volume) measure on 
the fields were meaningful, where the integration variable $\phi$ runs 
over the distributional fields. 
It is not difficult to show that the function 
$\e^{-V}$ is Borel-measurable on the space of smooth fields, and is bounded. 
One is tempted thus to define the Feynman measure $\mu$ of the interacting theory 
to be the product of the density function $\e^{-V}$ and the Gaussian measure $\gamma_{T}$, meaning that 
$\int(\dots)\,\mathrm{d}\mu(\phi):=\int(\dots)\,\e^{-V(\phi)}\,\mathrm{d}\gamma_{T}(\phi)$ by 
the tentative definition. 
The well-known obstacle to this attempt is the fact that $\gamma_{T}$ lives on the space of distributional fields, 
whereas $\e^{-V}$ can only be evaluated on the space of function sense fields, 
since the interaction term contains spacetime integrals of point-localized products of 
fields. In order to bring $\e^{-V}$ and $\gamma_{T}$ to common grounds, 
one needs to bring the measure $\gamma_{T}$ to the space of function sense fields. 
This naturally forces one to introduce the notion of Wilsonian regularized Feynman functional 
integral. Namely, one needs to take some \emph{coarse-graining} operator $C$, which 
is a continuous linear map from the distributional fields to the smooth function 
sense fields.\footnote{Over affine spacetime, if one requires $C$ to respect 
the translation invariance, it will simply be a convolution operator by a test function. 
Equivalently, it corresponds to a UV damping in momentum space, as Wilson and contemporaries originally formulated. 
On manifolds, the precise notion and properties of coarse-grainig operators will be recalled in Section~\ref{secWilson}.}
The image space $\Ran(C)$ of $C$ corresponds to a space of UV damped fields, 
which is by construction, some subspace of the smooth function sense fields. 
The pushforward of $\gamma_{T}$ by $C$, denoted by $C_{*}\gamma_{T}$, is a 
finite Borel measure on $\Ran(C)$. Thus, the function $\e^{-V}$ will be 
integrable against this Wilsonian regularized Gaussian measure $C_{*}\gamma_{T}$, 
and therefore the product $\e^{-V}\,C_{*}\gamma_{T}$ meaningfully defines a 
finite Borel measure on $\Ran(C)$. 
That is the \emph{Wilsonian regularized Feynman measure} for the interacting theory, 
at a fixed regularization. Having pinned down this notion, 
given a family $(V_{C})_{C\in\{\mathrm{coarse{\shortminus}grainings}\}}$ of interaction terms 
one can define the corresponding family 
$(\mu_{C})_{C\in\{\mathrm{coarse{\shortminus}grainings}\}}$ of Wilsonian regularized interacting 
Feynman measures, by setting $\mu_{C}:=\e^{-V_{C}}\,C_{*}\gamma_{T}$. 
Such a family is then called a \emph{Wilsonian RG flow reaching out to all UV regularization 
strengths}\footnote{One could also formalize the notion of Wilsonian RG flows which do not 
reach out to all UV regularization strengths. These can be important for 
encoding surely non-renormalizable QFT models. In this paper, however, 
we do not address the mathematical theory of these.} 
whenever there exists a real valued functional $z$ of coarse-grainings, such that 
for all coarse-grainings $C$, $C'$, $C''$ satisfying $C''=C'\,C$, one has that 
the measure $z(C'')_{*}\,\mu_{C''}$ is the pushforward of the measure 
$z(C)_{*}\,\mu_{C}$ by $C'$, where $z(C)_{*}$ and $z(C'')_{*}$ denote the pushforward 
by the field rescaling operation by the real numbers $z(C)$ and $z(C'')$. 
The functional $z$ is called the running wave function renormalization 
factor.\footnote{The wave function renormalization factor $z$ has to be invoked flavor 
sectorwise, if the field theory is based on particle fields composed of 
multiple flavor sectors.} 
The measures $z(C)_{*}\,\mu_{C}$ and $z(C'')_{*}\,\mu_{C''}$ are nothing but the 
Wilsonian regularized interacting Feynman measures re-expressed on the rescaled fields. 
The intermediary pushforward by $C'$ is the rigorous formulation of ``integrating out'' 
intermediate frequency modes between $C$ and $C''$. That is, in a 
Wilsonian RG flow one proceeds from the UV toward the infrared (IR) 
by applying subsequent coarse-graining operators. 
A less formalism-heavy equivalent definition is the following:
\begin{eqnarray}
 \exists\,\mathrm{real\; valued\; functional }\; z \;\mathrm{ of\; coarse{\shortminus}grainings}: \cr
 \Big. \quad \forall\,\mathrm{coarse{\shortminus}grainings }\; C,C',C'' \;\mathrm{ with }\; C''=C'\,C : \cr
 {}_{\Bigg.} \qquad \forall\,\mathrm{real\; valued\; functional\; (``observable")\; } O \;\mathrm{of\; smooth\; fields} \,: \cr
 \qquad \int\limits_{\varphi''\in \Ran(C'')} O({\textstyle\frac{1}{z(C'')}}\varphi'') \;\mathrm{d}\mu_{C''}(\varphi'') \;=\; \int\limits_{\varphi\in \Ran(C)} O(C'\,{\textstyle\frac{1}{z(C)}}\varphi) \;\mathrm{d}\mu_{C}(\varphi)
\label{eqmuzC}
\end{eqnarray}
holds. 
An RG flow of Feynman measures can be equivalently described via their 
Fourier transforms, being the usual partition function
\begin{eqnarray}
 Z_{C}(J) := \int\limits_{\varphi\in \Ran(C)} \e^{\I\,(J\vert\varphi)} \;\mathrm{d}\mu_{C}(\varphi) 
 = \int\limits_{\varphi\in \Ran(C)} \e^{\I\,(J\vert\varphi)} \;\e^{-V_{C}(\varphi)}\; \mathrm{d}(C_{*}\gamma_{T})(\varphi) \cr
 \qquad\qquad\qquad\qquad\qquad \left[ = \int\limits_{\varphi\in \Ran(C)} \e^{\I\,(J\vert\varphi)} \;\e^{-(T_{C}+V_{C})(\varphi)}\; \mathrm{d}\varphi \right] ,
\label{eqZC}
\end{eqnarray}
where $J$ runs over the compactly supported distributions (``currents''), and the 
expression in the square brackets is the customary informal presentation, as if a Lebesgue (volume) measure on 
$\Ran(C)$ were meaningful. 
The Wilsonian RGE in terms of the partition function reads as 
$Z_{C''}({\textstyle\frac{1}{z(C'')}}\,J) \;=\; Z_{C}(C'^{\,t}\,{\textstyle\frac{1}{z(C)}}\,J)$, 
referring to the notations of Eq.(\ref{eqmuzC}), where $C'^{\,t}$ denotes the transpose of $C'$. 
Finally, when re-expressed in terms of moments, the Wilsonian RGE reads
\begin{eqnarray}
 \Big.\exists\,\mathrm{real\; valued\; functional }\; z \;\mathrm{ of\; coarse{\shortminus}grainings}: \cr
 \quad \forall\,\mathrm{coarse{\shortminus}grainings }\; C,C',C'' \;\mathrm{ with }\; C''=C'\,C\,: \cr
 \qquad \bigg. z(C'')^{n}\;\mathcal{G}^{(n)}_{C''}\,\;=\;\,z(C)^{n}\;{\otimes}^{n}C'\,\mathcal{G}^{(n)}_{C} \quad(n=0,1,2,\dots).
\label{eqGzC}
\end{eqnarray}
Here, for any given coarse-graining $C$ the symbol 
$\mathcal{G}_{C}:=(\mathcal{G}_{C}^{(0)},\mathcal{G}_{C}^{(1)},{\dots})$ denotes 
the collection of moments of the Wilsonian regularized Feynman measure 
$\mu_{C}$, 
moreover $\otimes^{n}C'\,\mathcal{G}_{C}^{(n)}$ means the application of $C'$ to each variable of 
$\mathcal{G}_{C}^{(n)}$. 
It follows immediately from Eq.(\ref{eqGzC}), that each moment $\mathcal{G}_{C}^{(n)}$ have to be smooth function of the $n$-fold copy 
of the spacetime manifold.

In arbitrary, e.g.\ Lorentzian signatures and in a generally covariant 
setting, genuine Feynman measures in the above sense are known to be problematic: 
rather the collection of formal moments, i.e.\ the Feynman type $n$-field correlators are taken 
as the fundamental object of interest. Their Wilsonian RG flows are formulated 
by requiring Eq.(\ref{eqGzC}), as a definition of the RGE. 
In this paper we prove two statements on the space of such flows. 
\emph{Statement~(A)}: over generic spacetime manifolds, the space of rescaled correlators 
$z(C)^{n}\,\mathcal{G}_{C}^{(n)}$ ($C\in\{\mathrm{coarse{\shortminus}grainings}\}$) 
of these flows form a topological vector space, 
which is Hausdorff, locally convex, complete, nuclear, semi-Montel and Schwartz. 
That is, they form a generalized function space having favorable properties similar 
to that of $n$-variate distributions. 
Quite evidently, the pertinent space of flows is nonempty, as for any fixed $n$-variate 
distribution $G^{(n)}$, the family defined by the ansatz
\begin{eqnarray}
 \mathcal{G}^{(n)}_{C}=z(C)^{-n}\,{\otimes^{n}C\,G^{(n)}} \qquad (\;C\in\{\mathrm{coarse{\shortminus}grainings}\}\;)
\label{eqAnsatz}
\end{eqnarray}
solves the RGE Eq.(\ref{eqGzC}). It is not evident however from first principles, 
that this ansatz would be exhaustive.\footnote{For instance, 
in the space of Colombeau generalized functions, 
the subspace corresponding to ordinary distributions is known not to saturate the full 
space. Colombeau generalized functions 
\cite{Colombeau1984} was an early attempt to formalize RG flows in perturbative renormalization theory.} 
The second main result of the paper, called \emph{statement~(B)}, 
is that the ansatz Eq.(\ref{eqAnsatz}) is in fact exhaustive for 
QFT models of bosonic fields over an affine (i.e.,\ flat) spacetime. 
Statement~(A) indicates that statement~(B) might be generically true, 
not only for bosonic fields and flat spacetime, but we were not yet able to 
construct a formal proof for that, therefore is worth for future investigations.

The factorization formula of statement~(B) also implies that, under the given conditions, 
the rescaled correlators can only have UV singularities which are at worst distributional, 
and that is rather non-evident to see directly from first principles. 
In QFT terms, one can phrase it like this: under the given 
conditions, a Wilsonian RG flow reaching out to all UV regularization 
strengths is nonperturbatively multiplicatively renormalizable, i.e.\ there 
exists some regularization-independent distributional correlator $G^{(n)}$ ($n{=}0,1,2,{\dots}$) 
such that Eq.(\ref{eqAnsatz}) holds. 
Strictly speaking, up to now, the existence of such distributional correlator 
describing the UV infinity end of an RG flow has only been shown for 
low dimensional toy models, such as sin--Gordon or sinh--Gordon models 
\cite{Frohlich1975,Frohlich1977,Smirnov1992,Lukyanov1997,Kozlowski2023}. 
Statement~(B) says that this penomenon is generic for QFT models admitting 
flows reaching out to all UV regularization strengths.\footnote{Knowing 
the factorization formula Eq.(\ref{eqAnsatz}) in advance can come 
useful when solving for the flow of correlators of a QFT model. The regularized 
correlators must solve the Wilsonian RGE Eq.(\ref{eqGzC}) along with 
the Wilsonian regularized master Dyson--Schwinger equation \cite{Laszlo2022}, 
as an equation of motion. Equivalently, in transformed variables, they must 
solve the better known Wetterich equation \cite{Angelo2024}. 
Statement~(B) transforms these functional PDE type equations to algebraic equations.}

The structure of the paper is as follows. In Section~\ref{secWilson} the mathematical 
definition of the coarse-graining operators and of the $n$-variate Wilsonian type generalized 
functions is recalled from \cite{Laszlo2022}, 
moreover statement~(A) is proved. 
In Section~\ref{secBF} statement~(B) is proved. 
In Section~\ref{secConclusion} the ramifications of these findings in QFT is discussed. 
The proofs heavily rely on the mathematical theory of topological vector 
spaces. Therefore, the paper is closed by \ref{secApp}, summarizing some important 
facts on the theory of distributions and topological vector spaces.

\section{Wilsonian type generalized functions}
\label{secWilson}

In this section, let us denote by $\M$ an arbitrary finite dimensional smooth 
orientable and oriented manifold with or without boundary, 
modeling a generic spacetime manifold. If with boundary, the so-called cone condition 
is assumed for it, so that the Sobolev and Maurin compact embedding theorems hold over 
local patches. 
Whenever $V(\M)$ is some finite dimensional real vector bundle over $\M$, the notation 
$V^{\times}(\M):=V^{*}(\M)\otimes \left(\wedge^{\dim(\M)}T^{*}(\M)\right)$ will be used 
for its densitized dual vector bundle. For two vector bundles $V(\M)$ and 
$U(\mathcal{N})$ over base manifolds $\M$ and $\mathcal{N}$, the notation $V(\M)\boxtimes U(\mathcal{N})$ 
will be used for their external tensor product, which is then a vector bundle over the base $\M\times\mathcal{N}$. 
The shorthand notation $\EEn$ and $\EExn$ shall be used for the 
smooth sections of $\boxtimes^{n}V(\M)$ and of $\boxtimes^{n}V^{\times}(\M)$ ($n\in\N_{0}$), respectively, 
with their canonical $\EE$ type smooth function topology. It is common knowledge 
that since the Sobolev and Maurin embedding theorems hold locally, these spaces 
are nuclear Fr\'echet (NF) spaces. 
Their corresponding topological strong dual spaces, 
denoted as usual by $\EEn'$ and $\EExn{}'$, are dual nuclear Fr\'echet (DNF) spaces, 
being the spaces of corresponding compactly supported distributions. The symbols 
$\DDn$ and $\DDxn$, as usual, will denote the corresponding compactly 
supported smooth sections (test sections), with their canonical $\DD$ type 
test function topology. These are known to be also NF spaces when $\M$ is 
compact, and if $\M$ is noncompact they are known to be countable strict inductive 
limit with closed adjacent images of NF spaces (also called LNF spaces), the inductive 
limit taken for an increasing countable covering by compact patches of $\M$. 
Their corresponding topological strong dual spaces, 
denoted as usual by $\DDn'$ and $\DDxn{}'$, are dual LNF (DLNF) spaces, 
being the spaces of corresponding distributions. 
One has the canonical continuous linear embeddings $\EEn\subset\DDxn{}'$ 
and $\DDn\subset\EExn{}'$. Rather obviously, we will use the shorthand 
$\EE=\EE_{1}$, $\DD=\DD_{1}$ etc, respectively.

\begin{remark}
The notion of coarse-graining operators is invoked as follows \cite{Hormander2007,Shubin2001,Radzikowski1996,Laszlo2022}.
\begin{enumerate}[(i)]
 \item A continuous linear map $C:\,\EEx{}'\rightarrow\EE$ is called a \defin{smoothing operator}. 
By means of the Schwartz kernel theorem over manifolds, there is a corresponding unique smooth section $\kappa$ of 
$V(\M)\boxtimes V^{\times}(\M)$, such that 
$\forall\,\varphi\in\DD,\,x\in\M:\;(C\, \varphi)(x)=\int_{y\in\M}\kappa(x,y)\,\varphi(y)$ holds. 
Thus, one may write $C_{\kappa}$ in order to emphasize this.
 \item A smoothing operator $C_{\kappa}$ is called \defin{properly supported} (or partially compactly supported), 
whenever for all compact $\K\subset\M$, the closure of the sets 
$\left\{(x,y)\in\M\times\M\,\vert\,x\in\K,\,\kappa(x,y)\neq0\right\}$ and $\left\{(x,y)\in\M\times\M\,\vert\,y\in\K,\,\kappa(x,y)\neq0\right\}$
are compact. 
A properly supported smoothing operator $C_{\kappa}$ can be considered 
as continuous linear operator 
$\DD\rightarrow\DD$, $\EE\rightarrow\EE$, $\EEx{}'\rightarrow\EEx{}'$, $\DDx{}'\rightarrow\DDx{}'$, 
moreover as continuous linear operator
$\EEx{}'\rightarrow\EE$, $\DDx{}'\rightarrow\EE$, 
$\EEx{}'\rightarrow\DD$, 
respectively. 
Moreover, one can construct the corresponding \defin{formal transpose kernel} $\kappa^{t}$, 
being a section of $V^{\times}(\M)\boxtimes V(\M)$, which will invoke a 
properly supported smoothing operator $C_{\kappa^{t}}$ when exchanging $V(\M)$ versus $V^{\times}(\M)$ in their role. 
The space of properly supported smoothing operators inherit the natural convergence vector space structure 
from the spaces $\DD$ and $\DDx$ (\cite{Laszlo2022}~Appendix~B). Therefore, one can speak about 
sequentially continuous maps going from the space of properly supported smoothing 
operators to other convergence vector spaces, e.g.\ to the reals. 
By construction, if $\M$ were an affine space, the convolution operator 
by a real valued test function would be a properly supported smoothing operator (with translationally invariant kernel).
 \item A properly supported smoothing operator $C_{\kappa}$ is called \defin{coarse-graining operator} 
and its kernel $\kappa$ a \defin{mollifying kernel} iff 
$C_{\kappa}:\,\EEx{}'\rightarrow\DD$ and $C_{\kappa^{t}}:\,\EE'\rightarrow\DDx$ are injective. 
For instance, if $\M$ were an affine space, then the convolution operator 
by a real valued nonzero test function would be a coarse-graining operator, since by 
means of the Paley--Wiener--Schwartz theorem (\cite{Hormander1990}~Theorem7.3.1) 
it is injective on the above spaces of compactly supported distributions.
\end{enumerate}
\end{remark}

The above notion of coarse-graining operator generalizes the notion of 
convolution operators by test functions on affine spaces to generic 
manifolds.

\begin{remark}
A natural partial ordering is present on coarse-graining operators \cite{Laszlo2022}.
\begin{enumerate}[(i)]
 \item Given two coarse-graining operators $C_{\kappa}$ and $C_{\lambda}$, it is 
said that $C_{\kappa}$ is \defin{less ultraviolet (UV) than} $C_{\lambda}$, in notation 
$C_{\kappa}\preceq C_{\lambda}$, iff $C_{\kappa}=C_{\lambda}$ or there exists 
a coarse-graining operator $C_{\mu}$ such that $C_{\kappa}=C_{\mu}\,C_{\lambda}$ holds. 
This relation by construction is reflexive and transitive. Moreover, it is natural in the sense that it is diffeomorphism invariant 
(or more precisely, it is invariant to $V(\M)\rightarrow V(\M)$ vector bundle automorphisms). 
In the case of affine $\M$, the pertinent relation is also natural on the space of 
convolution operators by test functions: it is invariant to the affine transformations of $\M$.
 \item In \cite{Laszlo2022}~Appendix~B it is shown that $\preceq$ is also antisymmetric, 
i.e.\ is a partial ordering. A rather direct proof can be also given to its antisymmetry 
in the special case of convolution operators on affine spaces, via restating the antisymmetry on 
the Fourier transforms, and using the Paley--Wiener--Schwartz theorem 
in combination with the Riemann--Lebesgue lemma (\cite{Blanchard2015}~Ch10.1~Lemma10.1).
\end{enumerate}
\end{remark}

In order to construct a proof for statement~(A) of Section~\ref{secIntro}, 
we now define the space of rescaled $n$-field correlators obeying the 
Wilsonian RGE Eq.(\ref{eqGzC}). 
Referring to the notations of Section~\ref{secIntro}, a rescaled correlator shall be the product 
$w^{(n)}_{C}:=z(C)^{n}\,\mathcal{G}^{(n)}_{C}$ ($C$ being a coarse-grainig), 
with $z$ and $\mathcal{G}$ obeying Eq.(\ref{eqGzC}). 
That is, the wave function renormalization factor is merged notationally into 
the regularized correlator, and then the space of these rescaled correlators 
will be studied. The formal definition goes as follows, with somewhat 
simplified notations.

\begin{definition}
Denote by $\CC$ the space of coarse-graining operators (or equivalently, of mollifying kernels), 
and let $n\in\N_{0}$. Then, the set of maps
\begin{eqnarray}
 W_{n} \;:=\; 
 \big\{\, {w{:}\,\CC{\rightarrow}\EEn} \,\big\vert\, \forall \kappa,\lambda\in \CC,\,\kappa\preceq\lambda\,\mathrm{(with}\,C_{\kappa}=C_{\mu}C_{\lambda},\,\mu\in\CC\mathrm{)}:\cr
        \qquad\qquad\qquad\qquad\qquad\qquad\qquad\qquad\qquad\qquad w(\kappa)={\otimes}^{n}C_{\mu}\,w(\lambda) \,\big\}
\end{eqnarray}
is called the space of \defin{$n$-variate Wilsonian generalized functions}.
\label{defWn}
\end{definition}

Clearly, the above definition formalizes the space of Wilsonian renormalization 
group flows of $n$-variate smooth functions, as outlined in Section~\ref{secIntro}.

\begin{theorem}
$W_{n}$ is a vector space over $\RR$.
There is a natural linear map
\begin{eqnarray}
 j:\quad \DDxn{}'\longrightarrow W_{n},\quad \;\omega\longmapsto\widehat{\omega},
 \qquad \mathrm{with }\;\, \widehat{\omega}(\kappa):={\otimes}^{n}C_{\kappa}\,\omega \quad (\forall\,\kappa\in\CC)
\label{eqjn}
\end{eqnarray}
which is injective. That is, the space of $n$-variate Wilsonian generalized functions is larger than $\{0\}$, and contains the $n$-variate distributions.
\label{thmjInj}
\end{theorem}

\begin{proof}
Only the injectivity of $j$ may not be immediately evident. That is seen by taking any $\omega\in\DDxn{}'$ and 
a sequence $\kappa_{i}$ ($i\in\N_{0}$) of mollifying kernels which are Dirac delta 
approximating. Then, the sequence of distributions 
$\otimes^{n}C_{\kappa_{i}}\,\omega$ ($i\in\N_{0}$) is convergent to $\omega$ in 
the weak-* topology. If $\omega$ were such that 
$\forall\,\kappa\in\mathscr{C}:\;\otimes^{n}C_{\kappa}\,\omega=0$ holds, then for an above kind of sequence 
$\forall\,i\in\N_{0}:\;\otimes^{n}C_{\kappa_{i}}\,\omega=0$ holds. Therefore, its 
weak-* limit, being equal to $\omega$, is zero. That is, $\omega=0$.
\end{proof}

The aim of the paper is to see if $W_{n}$ is strictly larger than $j[\DDxn{}']$ or not.

\begin{remark}
$W_{n}$ can naturally be topologized as follows. Recall that the space of 
coarse-grainings $(\CC,\preceq)$ was a partially ordered set, and that 
by construction, for all $C_{\kappa},C_{\lambda}\in\CC$ and $C_{\kappa}\preceq C_{\lambda}$ 
there existed a unique continuous linear map $F_{\lambda,\kappa}:\,\EE\rightarrow\EE$ such that 
$C_{\kappa}=F_{\lambda,\kappa}\,C_{\lambda}$ holds. In addition, for 
all $C_{\kappa},C_{\lambda},C_{\mu}\in\CC$ and $C_{\kappa}\preceq C_{\lambda}\preceq C_{\mu}$ 
the corresponding maps 
satisfy $\otimes^{n}F_{\mu,\kappa}=\otimes^{n}F_{\lambda,\kappa}\circ \otimes^{n}F_{\mu,\lambda}$. 
Therefore, the pair 
$\left(\; \left( \EEn \right)_{\kappa\in\CC} \;, \; \left( {\otimes}^{n}F_{\lambda,\kappa} \right)_{\kappa,\lambda\in\CC\;\mathrm{and}\;\kappa\preceq\lambda} \;\right)$
forms a projective system (see also e.g.\ \cite{DeJong2021}~Ch4.21). 
It is seen that $W_{n}$ is the projective limit of the above projective 
system.\footnote{Note that some pieces of literature 
require the partially ordered index set to be forward directed, but this is 
not necessary for the projective limit to be meaningful, see also \cite{DeJong2021}~Ch4.21.} 
The canonical projections are $\big( \Pi_{\kappa} \big)_{\kappa\in\CC}$ with 
$\Pi_{\kappa}:W_{n}\rightarrow\EEn,w\mapsto w(\kappa)$ (for all $\kappa\in\CC$). 
$W_{n}$ can be endowed with the natural projective limit vector topology, 
being the Tychonoff topology, i.e.\ the weakest topology such that the 
canonical projection maps are continuous.
\end{remark}

The following general result can be stated on the topology of $W_{n}$.

\begin{theorem}
The projective limit vector topology on $W_{n}$ exists, and has the properties:
\begin{enumerate}[(i)]
 \item \label{thmProjLimi} It is Hausdorff, locally convex, nuclear, complete.
 \item \label{thmProjLimii} It is semi-Montel, and thus semi-reflexive.
 \item \label{thmProjLimiii} It has the Schwartz property.
\end{enumerate}
\label{thmProjLim}
\end{theorem}

\begin{proof} We deduce these from the permanence properties of the projective limit.

(\ref{thmProjLimi}) First of all, the projective limit topology on a projective system of 
topological vector spaces 
exists and is a vector topology, see remark (i) after \cite{Treves1970}~Proposition50.1.
Moreover, all the spaces in $\big( \EEn \big)_{\kappa\in\CC}$ are Hausdorff and for all 
$w\in W_{n}{\setminus}\{0\}$ there is at least one $\kappa\in\CC$ such that 
$\Pi_{\kappa}w\neq 0$, by definition. Therefore, by means of the same remark, 
the pertinent topology is Hausdorff. 
All the spaces in the projective system are locally convex, therefore by 
means of the same remark, the projective limit topology is also locally convex. 
By means of \cite{Treves1970}~Proposition50.1~(50.7), the Hausdorff projective limit 
respects nuclearity, therefore $W_{n}$ is nuclear. Completeness is 
also a simple consequence of the completeness of each space in the system 
$\big( \EEn \big)_{\kappa\in\CC}$, see \cite{Schaefer1999}~ChII~5.3.

(\ref{thmProjLimii}) The semi-Montel property is a consequence of the Montel (and thus, semi-Montel) property of each space in the 
system $\big( \EEn \big)_{\kappa\in\CC}$ and of \cite{Horvath1966}~Ch3.9~Proposition6 and \cite{Horvath1966}~Ch3.9~Exercise3. 
It is semi-reflexive since it is semi-Montel \cite{Horvath1966}~Ch3.9~Proposition1. 
(See also \cite{Horvath1966}~p.442~Table3.)

(\ref{thmProjLimiii}) Schwartz property follows from \cite{Horvath1966}~Ch3.15~Proposition6(c).
\end{proof}

The above theorem proves statement~(A) in Section~\ref{secIntro}. 
As seen, the topological vector space $W_{n}$ has rather similar properties 
to the space of ordinary distributions $\DDxn{}'$. 
One may conjecture that $j[\DDxn{}']\subset W_{n}$ saturates $W_{n}$. 
For the generic case, we were unable to construct a proof for this claim. 
However, for the special case of bosonic fields over affine spaces (flat spacetime), 
this surjectivity property is proved in the following section.

\section{The symmetrized case over affine space}
\label{secBF}

In this section, denote by $\M$ a finite dimensional real affine space, with 
subordinate vector space (``tangent space'') $T$.\footnote{Without 
loss of generality, one may even take $\M:=T:=\RR^{N}$ for some $N\in\N_{0}$.} 
In such scenario, due to the existence of an affine-constant nonvanishing maximal 
form field (corresponding to the Lebesgue measure), one does not need to distinguish 
$V^{\times}(\M)$ from $V^{*}(\M)$, since one may use the identification 
$\wedge^{\dim(\M)}T^{*}\equiv\RR$, up to a real multiplier. The smooth 
sections of a trivialized vector bundle $V(\M)$ can be identified with $\M\rightarrow V$ smooth functions, 
$V$ being the typical fiber. For simplicity of notation, in this section only 
scalar valued fields, i.e.\ $V=\RR$ are considered. The generic vector valued 
case can be recovered straightforwardly, \emph{mutatis mutandis}.

Due to affine base manifold and trivialized bundles over it, the notion of 
convolution operators by real valued test functions is meaningful. Given 
$f\in\DD$, the convolution operator acts as $C_{f}:\,\DD\rightarrow\DD$ with $C_{f}\,g:=f\star g$ ($\forall\,g\in\DD$) 
using the traditional star notation. Such a convolution operator $C_{f}$ is 
a coarse-grainig operator in terms of Section~\ref{secWilson}, with 
affine-translationally invariant mollifying kernel. All the previously mentioned 
properties hold for it, and in addition, it is commutative, i.e.\ 
$C_{g}\,C_{f}=C_{f}\,C_{g}$ ($\forall\,f,g\in\DD$). In some of the proofs this 
special property will be relied on. Clearly, the relation $\preceq$ can be 
restricted onto the space $\DDz$, and the definition of $W_{n}$ may be reformulated 
in case of affine spaces using the partially ordered set $(\DDz,\preceq)$ in 
Definition~\ref{defWn} instead of generic coarse-graining operators.

In this section, only bosonic fields are considered. Therefore, the notation 
$\EEvn$ and $\DDvn$ are introduced for the totally symmetrized subspace of 
$\EEn$ and $\DDn$, respectively, with their corresponding totally 
symmetrized distributions $\EEvn{}'$ and $\DDvn{}'$. The topological vector 
space of $n$-variate totally symmetric Wilsonian renormalization group flows 
$\Wvn$ can be also introduced based on Definition~\ref{defWn}, stated below.

\begin{definition}
Let $n\in\N_{0}$. Then, the set of maps
\begin{eqnarray}
 \Wvn \;:=\; 
 \big\{\, {w{:}\,{\DDz}{\rightarrow}\EEvn} \;\big\vert\; \forall f,g\in\DDz,\,f\preceq g \;\mathrm{ (with }\; f=C_{h}g\mathrm{)}:\,\cr
  \qquad\qquad\qquad\qquad\qquad\qquad\qquad\qquad\qquad\qquad w(f)={\otimes}^{n}C_{h}\,w(g) \;\big\}
\end{eqnarray}
is called the space of \defin{$n$-variate symmetric Wilsonian generalized functions}. 
\end{definition}

Clearly, the analogy of Theorem~\ref{thmProjLim} applies to $\Wvn$. 
Also, the natural continuous linear injection 
$j:\,\DDvn{}'\rightarrow\Wvn$ can be defined, in the analogy of Theorem~\ref{thmjInj}. 
The aim of this section is to prove that this canonical injection map $j$ is surjective. 
For this purpose, one needs to invoke a number of tools, as follows. 
First, recall the polarization identity for totally symmetric $n$-forms.

\begin{lemma}[polarization identity for $n$-forms, see also \cite{Thomas2014}~formula~A.1]
Let $V$ and $W$ be real or complex vector spaces and 
$u:\,V\rightarrow W$ be an $n$-order homogeneous polynomial. 
Then, the map
\begin{eqnarray}
 u^{\vee}:\quad {\times}^{n}V\longrightarrow W,\quad (x_{1},\dots,x_{n})\longmapsto u^{\vee}(x_{1},\dots,x_{n}) \;:=\cr
 \qquad\qquad \frac{1}{n!}\sum_{\epsilon_{1}=0,\dots,\epsilon_{n}=0}^{1} (-1)^{n-(\epsilon_{1}+\dots+\epsilon_{n})}\quad u(\epsilon_{1}x_{1}+\dots+\epsilon_{n}x_{n})
\end{eqnarray}
is an $n$-linear symmetric map, moreover $\forall x\in V:\;u^{\vee}(x,\dots,x)=u(x)$ holds.
\label{lemPolar}
\end{lemma}

The polarization identity motivates the definition of the symmetrized convolution. 
For fixed $f_{1},\dots,f_{n}\in\DD$, set
\begin{eqnarray}
 C^{\vee}_{f_{1},\dots,f_{n}} \;:=\; \frac{1}{n!} \sum_{\epsilon_{1}=0,\dots,\epsilon_{n}=0}^{1} (-1)^{n-(\epsilon_{1}+\dots+\epsilon_{n})}\quad {\otimes}^{n}C_{\epsilon_{1}f_{1}+\dots+\epsilon_{n}f_{n}}
\end{eqnarray}
which is then a linear operator between the function spaces of the domain and range of 
$C_{f_{1},\dots,f_{n}}:=C_{f_{1}}\otimes\dots\otimes C_{f_{n}}=C_{f_{1}\otimes\dots\otimes f_{n}}$, 
with the same properties. 
Moreover, $C^{\vee}_{f_{1},\dots,f_{n}}$ is $n$-linear and symmetric in its parameters 
$f_{1},\dots,f_{n}\in\DD$ and one has the identity $C^{\vee}_{f,\dots,f}=C_{f,\dots,f}$. 
Quite naturally, one has the identity 
$C^{\vee}_{f_{1},\dots,f_{n}}=\frac{1}{n!}\sum_{\pi\in\Pi_{n}}C_{f_{\pi(1)},\dots,f_{\pi(n)}}$ as well, 
with $\Pi_{n}$ denoting the set of permutations of the index set $\{1,\dots,n\}$. 
Furthermore, $C^{\vee}_{f_{1},\dots,f_{n}}=C_{\Sym(f_{1}\otimes\dots\otimes f_{n})}$ holds, 
where $\Sym(f_{1}\otimes\dots\otimes f_{n}):=\frac{1}{n!}\sum_{\pi\in\Pi_{n}}f_{\pi(1)}\otimes\dots\otimes f_{\pi(n)}\in\DDvn\subset\DDn$.

\begin{definition}
Take the canonical projection operators $\big(\Pi_{f}\big)_{f\in\DDz}$ 
from the projective system defining $\Wvn$. These act as 
$\Pi_{f}w:=w(f)$ on each $w\in \Wvn$ ($\forall f\in\DDz$) and 
extend this notation, for convenience, by $\Pi_{f}w:=0$ whenever $f=0$. 
Then, for all $f_{1},\dots,f_{n}\in\DD$, the following map is defined:
\begin{eqnarray}
 \Piv_{f_{1},\dots,f_{n}}:\; \Wvn\longrightarrow\EEvn,\; w\longmapsto \Piv_{f_{1},\dots,f_{n}}w \;:= \cr
 \qquad\qquad\qquad\qquad \frac{1}{n!}\sum_{\epsilon_{1}=0,\dots,\epsilon_{n}=0}^{1} (-1)^{n-(\epsilon_{1}+\dots+\epsilon_{n})}\quad \Pi_{{}_{\epsilon_{1}f_{1}+\dots+\epsilon_{n}f_{n}}}w
\end{eqnarray}
which may be called the \defin{polarized version of the canonical projection}.
\end{definition}

By construction, for all $\omega\in\DDvn{}'$, one has that 
$\forall\,f_{1},\dots,f_{n}\in\DD:\;\Piv_{f_{1},\dots,f_{n}}\widehat{\omega}=C^{\vee}_{f_{1},\dots,f_{n}}\omega$ holds, 
which is the rationale behind the above definition.
In addition, for all $f_{1},\dots,f_{n}\in\DD$ and $\omega\in\DDn'$ one has the identity 
$\big(\Piv_{f_{1}^{t},\dots,f_{n}^{t}}\,\widehat{\omega}\big)(0)=\big(C^{\vee}_{f_{1}^{t},\dots,f_{n}^{t}}\,\omega\big)(0)=(\Sym(\omega)\,\vert\,f_{1}\otimes\dots\otimes f_{n})$, 
where $\Sym(\omega)$ is the totally symmetrized part of $\omega$, and 
$f^{t}$ is the reflected version of $f$. 
This motivates the construction of the tentative inverse map of $j$, below.

\begin{definition}
Denote by $\Map(A,B)$ the set of $A\rightarrow B$ maps between sets $A,\,B$.
Using this notation, invoke the linear map
\begin{eqnarray}
 \ell:\quad \Wvn\longrightarrow \Map({\times}^{n}\DD,\EEvn),\quad w\longmapsto\mathring{w}, \cr
 \Bigg. \mathrm{with }\; \mathring{w}(f_{1},\dots,f_{n}):=\Piv_{f_{1}^{t},\dots,f_{n}^{t}}w \quad (\mathrm{for\; any }\; \,f_{1},\dots,f_{n}\in\DD).
\end{eqnarray}
Using that, invoke the linear map
\begin{eqnarray}
 k:\quad \Wvn\longrightarrow \Map({\times}^{n}\DD,\RR),\quad w\longmapsto\widetilde{w}, \cr
 \Bigg. \mathrm{with }\; \widetilde{w}(f_{1},\dots,f_{n}):=\big(\mathring{w}(f_{1},\dots,f_{n})\big)(0) \quad (f_{1},\dots,f_{n}\in\DD).
\end{eqnarray}
This map $k$ will be the \defin{tentative inverse} of the continuous linear injection $j$.
\end{definition}

First, we show that for all $w\in\Wvn$, the map $\widetilde{w}:\,{\times}^{n}\DD\rightarrow\RR$ is $n$-linear in its arguments.

\begin{lemma}
For all $w\in\Wvn$, the map $\mathring{w}:\,{\times}^{n}\DD\rightarrow\EEvn$ is linear in each variable and is totally symmetric. 
The map $\widetilde{w}:\,{\times}^{n}\DD\rightarrow\RR$ is also linear in each variable and totally symmetric.
\end{lemma}

\begin{proof}
By the definition of $\Wvn$, one has that for all $g,f_{1},\dots,f_{n}\in\DD$ and $\alpha\in\RR$, 
\begin{eqnarray}
 ({\otimes}^{n}C_{g})\,\Piv_{\alpha f_{1},\dots,f_{n}}w = \Piv_{C_{g}\alpha f_{1},\dots,C_{g}f_{n}}w
\end{eqnarray}
which due to the commutativity of convolution further equals to
\begin{eqnarray}
 \Piv_{C_{\alpha f_{1}}g,\dots,C_{f_{n}}g}w = C^{\vee}_{\alpha f_{1},\dots,f_{n}}\,\Piv_{g,\dots,g}w = \alpha\, C^{\vee}_{f_{1},\dots,f_{n}}\Piv_{g}w \cr
 \Big. = \alpha\, \Piv_{C_{f_{1}}g,\dots,C_{f_{n}}g}w
\end{eqnarray}
which again due to the commutativity of convolution further equals to
\begin{eqnarray}
 \alpha\, \Piv_{C_{g}f_{1},\dots,C_{g}f_{n}}w = \alpha\;({\otimes}^{n}C_{g})\,\Piv_{f_{1},\dots,f_{n}}w.
\end{eqnarray}
That is, $\forall g\in\DD:\,{\otimes}^{n}C_{g}\big(\Piv_{\alpha f_{1},\dots,f_{n}}w-\alpha\,\Piv_{f_{1},\dots,f_{n}}w\big)=0$. 
By \ref{secApp}~Lemma~\ref{lemLagrange}, this implies that $\Piv_{\alpha f_{1},\dots,f_{n}}w-\alpha\,\Piv_{f_{1},\dots,f_{n}}w=0$ holds.

One can prove in a completely analogous way that $\Piv_{f_{1}+f_{1}',\dots,f_{n}}w=\Piv_{f_{1},\dots,f_{n}}w+\Piv_{f_{1}',\dots,f_{n}}w$  for all $f_{1},f_{1}',f_{2},\dots,f_{n}\in\DD.$ Hence the map $(f_{1},\dots,f_{n})\mapsto\Piv_{f_{1},\dots,f_{n}}w$ 
is linear in its first, and rather obviously, in each of its variables.

Since the reflection map $f\mapsto f^{t}$ is linear, it also implies that the map 
$\mathring{w}:\,{\times}^{n}\DD\rightarrow\EEvn$ is linear in each of its variables. 
The evaluation map 
$\EEvn\rightarrow\RR,\phi\mapsto\phi(0)$ is linear, therefore it follows that the map 
$\widetilde{w}:\,{\times}^{n}\DD\rightarrow\RR$ is linear in each of its variables.

The total symmetry of $\widetilde{w}$ is by construction evident.
\end{proof}

\begin{remark}
For any $w\in \Wvn$ and corresponding $n$-linear map $\widetilde{w}:\,\times^{n}\DD\rightarrow\RR$, 
its \defin{linear form} 
$\widetilde{\underline{w}}:\,{\otimes}^{n}\DD\rightarrow\RR$ can be defined 
to be the unique linear map for which 
$\widetilde{\underline{w}}(f_{1}\otimes\dots\otimes f_{n})=\widetilde{w}(f_{1},\dots,f_{n})$ holds 
($\forall\,f_{1},\dots,f_{n}\in\DD$). Due to the total symmetry of 
$\widetilde{w}$, the linear map $\widetilde{\underline{w}}$ is totally symmetric.
\end{remark}

Now we show that for any $w\in \Wvn$ the linear map 
$\widetilde{\underline{w}}:\,{\otimes}^{n}\DD\rightarrow\RR$ 
uniquely extends to a distribution.

\begin{lemma}
For all $w\in W$, there exists a unique distribution 
$\overline{\widetilde{w}}:\,\DDvn\rightarrow\RR$, 
such that for all $f_{1},\dots,f_{n}\in\DD$ the identity 
$(\overline{\widetilde{w}}\,\vert\,f_{1}\otimes\dots\otimes f_{n})=\widetilde{w}(f_{1},\dots,f_{n})$ 
holds. That is, $\widetilde{\underline{w}}:\,{\otimes}^{n}\DD\rightarrow\RR$ uniquely extends to the pertinent totally symmetric distribution.
\end{lemma}

\begin{proof}
Fix a $w\in\Wvn$, and define its corresponding symmetric linear map $\widetilde{\underline{w}}:\,{\otimes}^{n}\DD\rightarrow\RR$. 
For all $g\in\DD$ and $f_{1},\dots,f_{n}\in\DD$, one has the identity
\begin{eqnarray}
 \Big. \widetilde{\underline{w}}(C_{g}f_{1}\otimes\dots\otimes C_{g}f_{n}) = \widetilde{w}(C_{g}f_{1},\dots,C_{g}f_{n}) = \big(\Piv_{(C_{g}f_{1})^{t},\dots,(C_{g}f_{n})^{t}}w\big)(0) \cr
 \Big. = \big(\Piv_{C_{f_{1}^{t}}g^{t},\dots,C_{f_{n}^{t}}g^{t}}w\big)(0),
\label{eqDistr1}
\end{eqnarray}
which further equals to
\begin{eqnarray}
 \big(C^{\vee}_{f_{1}^{t},\dots,f_{n}^{t}}\Piv_{g^{t},\dots,g^{t}}w\big)(0) = ( \Piv_{g^{t},\dots,g^{t}}w \,\vert\, f_{1}\otimes\dots\otimes f_{n}),
\label{eqDistr2}
\end{eqnarray}
where the totally symmetric function $\Piv_{g^{t},\dots,g^{t}}w\in\EEn$ was regarded as a distribution. 
Moreover, due to the commutativity of convolution, the right hand side of Eq.(\ref{eqDistr1}) further equals to
\begin{eqnarray}
 \big(\Piv_{C_{g^{t}}f_{1}^{t},\dots,C_{g^{t}}f_{n}^{t}}w\big)(0) = \big({\otimes}^{n}C_{g^{t}}\,\Piv_{f_{1}^{t},\dots,f_{n}^{t}}w\big)(0).
\label{eqDistr3}
\end{eqnarray}
In total, one arrives at the identity
\begin{eqnarray}
 \forall\,f_{1}\otimes\dots\otimes f_{n}\in{\otimes}^{n}\DD: \cr
 \Big. \qquad\qquad ( \Piv_{g^{t},\dots,g^{t}}w \,\vert\, f_{1}\otimes\dots\otimes f_{n}) = \big({\otimes}^{n}C_{g^{t}}\,\Piv_{f_{1}^{t},\dots,f_{n}^{t}}w\big)(0),
\label{eqDistr4}
\end{eqnarray}
for given $g\in\DD$. 
Take a Dirac delta approximating sequence $g_{i}\in\DD$ ($i\in\N_{0}$), then 
from Eq.(\ref{eqDistr4}) it follows that the sequence of totally symmetric distributions 
$( \Piv_{g_{i}^{t},\dots,g_{i}^{t}}w \,\vert\, \cdot)\in\DDn'$ ($i\in\N_{0}$) 
is pointwise convergent on the subspace ${\otimes}^{n}\DD\subset\DDn$. 
\ref{secApp}~Lemma~\ref{lemAlgTensB} then implies 
that there exists a unique totally symmetric distribution 
$\overline{\widetilde{w}}\in\DDn'$, such that the sequence of totally symmetric distributions 
$\big(( \Piv_{g_{i}^{t},\dots,g_{i}^{t}}w \,\vert\, \cdot)-(\overline{\widetilde{w}}\,\vert\,\cdot)\big)\in\DDn'$ ($i\in\N_{0}$) 
converges to zero pointwise on the full $\DDn$. Moreover, Eq.(\ref{eqDistr4}) implies that 
$(\overline{\widetilde{w}}\,\vert\,f_{1}\otimes\dots\otimes f_{n})=\widetilde{w}(f_{1},\dots,f_{n})$ 
holds for all $f_{1},\dots,f_{n}\in\DD$, and therefore also 
$(\overline{\widetilde{w}}\,\vert\,f_{1}\otimes\dots\otimes f_{n})=\widetilde{\underline{w}}(f_{1}\otimes\dots\otimes f_{n})$ 
holds.
\end{proof}

\begin{remark}
The linear map $k:\,\Wvn\rightarrow\Map({\times}^{n}\DD,\RR)$ can be considered as distribution valued, 
i.e.\ the notation
\begin{eqnarray}
 k:\quad \Wvn\longrightarrow \DDvn{}',\quad w\longmapsto\widetilde{w}
\end{eqnarray}
is justified, via identifying $\widetilde{w}$ and $\widetilde{\underline{w}}$ and $\overline{\widetilde{w}}$.
\end{remark}

We are now in position to state and prove the main result of the paper, statement~(B) in Section~\ref{secIntro}. 
Roughly speaking, it says that symmetric Wilsonian generalized functions are in fact nothing more than distributions.

\begin{theorem}
The distribution valued linear map
\begin{eqnarray}
 k:\quad \Wvn\longrightarrow \DDvn{}',\quad w\longmapsto \widetilde{w}
\end{eqnarray}
is the inverse of the natural continuous linear injection
\begin{eqnarray}
 j:\quad \DDvn{}'\longrightarrow \Wvn,\quad \omega\longmapsto \widehat{\omega}.
\end{eqnarray}
\end{theorem}

\begin{proof}
Let $\omega\in\DDvn{}'$. Then, for all $f_{1},\dots,f_{n}\in\DD$ the identity
\begin{eqnarray}
 \Big. \big(k(j(\omega)) \,\big\vert\, f_{1}\otimes\dots\otimes f_{n}\big) = \big(k(\widehat{\omega}) \,\big\vert\, f_{1}\otimes\dots\otimes f_{n}\big) = \big(\Piv_{f_{1}^{t},\dots,f_{n}^{t}}\,\widehat{\omega}\big)(0) \cr
 \Big. = \big(C^{\vee}_{f_{1}^{t},\dots,f_{n}^{t}}\,\omega\big)(0) = (\omega\,\vert\,f_{1}\otimes\dots\otimes f_{n})
\end{eqnarray}
holds. This implies that the distributions $k(j(\omega))$ and $\omega$ coincide 
on the dense subspace ${\otimes}^{n}\DD\subset\DDn$, and therefore $k(j(\omega))=\omega$.

Let $w\in W$. Then, for all $g\in\DD$ and $f_{1},\dots,f_{n}\in\DD$, the smooth function 
$\Piv_{f_{1},\dots,f_{n}}j(k(w))\in\EEvn$ can be also regarded as a distribution, 
and one has the identity
\begin{eqnarray}
 \Big. \big( \Piv_{f_{1},\dots,f_{n}}j(k(w)) \,\big\vert\, {\otimes}^{n}g^{t} \big) = \big( \Piv_{f_{1},\dots,f_{n}}j(\widetilde{w}) \,\big\vert\, {\otimes}^{n}g^{t} \big) = \big( C^{\vee}_{f_{1},\dots,f_{n}}\widetilde{w} \,\big\vert\, {\otimes}^{n}g^{t} \big) \cr
 \Big. = \big( \widetilde{w} \,\big\vert\, C^{\vee}_{f_{1}^{t},\dots,f_{n}^{t}} ({\otimes}^{n}g^{t}) \big) = \big( \widetilde{w} \,\big\vert\, \Sym(C_{f_{1}^{t}}g^{t}\otimes\dots\otimes C_{f_{n}^{t}}g^{t}) \big) \cr
 \Big. = \big( \Piv_{(C_{f_{1}^{t}}g^{t})^{t},\dots,(C_{f_{n}^{t}}g^{t})^{t}} w \big)(0) = \big( \Piv_{C_{g}f_{1},\dots,C_{g}f_{n}} w \big)(0) \cr
 \Big. = \big( {\otimes}^{n}C_{g}\, \Piv_{f_{1},\dots,f_{n}} w \big)(0) = \big( \Piv_{f_{1},\dots,f_{n}}w \,\big\vert\, {\otimes}^{n}g^{t} \big)
\end{eqnarray}
where in the last two terms the smooth function $\Piv_{f_{1},\dots,f_{n}}w\in\EEvn$ was regarded as a distribution. 
Since $\Span\big\{ {\otimes}^{n}g^{t}\in\DDvn \,\big\vert\, g\in\DD \big\}$ separates points 
for totally symmetric smooth functions (\ref{secApp}~Lemma~\ref{lemLagrange}), it follows that 
for all $f_{1},\dots,f_{n}\in\DD$ the identity $\Piv_{f_{1},\dots,f_{n}}j(k(w))=\Piv_{f_{1},\dots,f_{n}}w$ holds, 
which implies $j(k(w))=w$.
\end{proof}

So far we have not said anything on whether the continuous bijection $j$ is a topological isomorphism between $\DDvn{}'$ and $\Wvn$, that is, whether its inverse map $k$ is continuous or not. Although we did not manage to answer this question, as a concluding result we show that  $k$ has certain weaker continuity properties.

\begin{theorem}
The distribution valued linear bijection
\begin{eqnarray}
 k:\quad \Wvn\longrightarrow\DDvn{}',\quad w\longmapsto\widetilde{w}
\end{eqnarray}
is continuous when the target space $\DDvn{}'$ is equipped with the weak dual topology against the 
subspace ${\otimes}^{n}\DD$. With the canonical topologies, $k$ is sequentially continuous.
\end{theorem}

\begin{proof}
Take a generalized sequence $w_{i}\in W_{n}$ ($i\in I$) such that it converges to 
$0$ in the $W_{n}$ topology. This implies that for all $f_{1},\dots,f_{n}\in\DD$ the generalized sequence 
$\Piv_{f_{1}^{t},\dots,f_{n}^{t}}w_{i}\in\EEvn$ ($i\in I$) converges to $0$ in the $\EEvn$ topology. Since the point evaluation map $\EEvn\rightarrow\RR$ is continuous, it follows that $\big(\widetilde{w}_{i}\,\big\vert\,f_{1}\otimes\dots\otimes f_{n}\big)\in\RR$ ($i\in I$) converges to $0$ in $\RR$. Hence  the generalized sequence $k(w_{i})\in\DDvn{}'$ ($i\in I$) converges to $0$ 
when the space $\DDvn{}'$ is equipped with the weak dual topology against 
${\otimes}^{n}\DD$, which proves the first statement of the theorem.

From the above, via applying \ref{secApp}~Lemma~\ref{lemAlgTensB}, the sequential 
continuity of $k$ follows when the target space is equipped with the weak-* topology. 
Then, using the Montel property of the space $\DDvn{}'$ it follows that 
the sequential continuity also holds when the target space is equipped with its 
canonical strong dual topology, which proves the second statement of the theorem.
\end{proof}

\begin{corollary}
We conclude that $\Wvn$ and $\DDvn{}'$ are isomorphic as convergence vector spaces.
\end{corollary}

\section{Concluding remarks}
\label{secConclusion}

In a QFT model, the vacuum state can be described by the Wilsonian renormalization 
group (RG) flow of the collection of the Feynman type $n$-field correlators ($n=0,1,2,\dots$). 
An RG flow is a parametric family of the collection of smoothed Feynman type 
correlators, the parameter being the strength of the UV regularization, and the 
instances with different UV regularization strengths are linked by the renormalization group equation (RGE). 
Important QFT models are those, which admit a flow meaningful at any UV regularization strength. 
Based on settings in which the Feynman measure genuinely exists, the 
distribution theoretically canonical definition of Wilsonian UV regularization was recalled: 
the UV regularization is most naturally implemented by coarse-gaining operators on the fields, 
where a coarse-graining is a kind of smoothing, analogous to convolution operator by a test function, 
i.e.\ to a momentum space damping. 
Using this notion of Wilsonian regularization, it was possible to mathematically rigorously and canonically define the space of the RG flows of 
correlators, even in a generally covariant and signature-independent setting (including Lorentzian). 
Quite naturally, flowing from the UV toward the IR means successive application of coarse-grainings after each-other, 
as seen in Eq.(\ref{eqGzC}).

It was shown that the space of coarse-graining operators admit a 
natural partial ordering, describing that one coarse-graining is less UV than 
an other. Recognizing this, the space of Wilsonian RG flows of rescaled field 
correlators reaching out to all UV regularization strengths was seen to form a projective limit space, 
made out of instances of smoothed field correlators. 
Using the known topological vector space 
properties of the smooth $n$-variate fields, and the known permanence properties 
of projective limit, the fundamental properties of the space of 
Wilsonian RG flows of rescaled correlators were established. 
That is the first main result of the paper, referred to as statement~(A): 
the flows of rescaled correlators form a topological 
vector space being Hausdorff, locally convex, complete, nuclear, semi-Montel, 
and Schwartz type space. That is, they form a generalized function space 
having many favorable properties similar to that of ordinary distributions. 
In addition, the ordinary distributional correlators can 
be naturally injected into that space by applying coarse-graining on its variables, 
i.e. via Eq.(\ref{eqAnsatz}).

It is quite natural to ask whether the above space of Wilsonian RG 
flows is much bigger than that of the subspace generated by the distributional correlators 
through Eq.(\ref{eqAnsatz}). 
The naive expectation would be that the former space is bigger than the latter one, since 
a Wilsonian RG flow is a more elaborate object in comparison to an ordinary distribution. 
Exotic UV behavior, more general than that of distributions, 
is also known to occur in other generalized function spaces, as it happens e.g.\ for the Colombeau generalized functions. 
The second main result of the paper, referred to as statement~(B), is that for bosonic fields over a flat (affine) spacetime manifold, 
the subspace generated by distributional correlators exhausts the space of 
Wilsonian RG flows of correlators. Moreover, with these conditions, 
these two spaces were found to be isomorphic in terms of their convergence 
vector space structures. Statement~(A) indicates that statement~(B) is 
likely to be generically true, 
not only for bosonic fields and flat spacetime. This conjecture 
is worth future investigations.

Physicswise, statement~(B) has the following meaning: 
for a QFT model based on bosonic fields over a flat (affine) spacetime manifold, 
the Wilsonian RG flow of Feynman type $n$-field correlators reaching out to all 
UV regularization strengths can always be legitimately factorized using the ansatz 
Eq.(\ref{eqAnsatz}), i.e.\ they are multiplicatively renormalizable. 
This factorization result is expected to come quite useful when attempting 
to solve the equation of motion of QFT for the RG flow of field correlators.\footnote{The equation of motion 
of QFT is the Wilsonian regularized master Dyson--Schwinger equation \cite{Laszlo2022} 
together with the RGE Eq.(\ref{eqGzC}). In different variables, these are 
equivalent to the better known Wetterich equation \cite{Angelo2024}. 
Since statement~(B) factors out a regularization-independent distributional 
correlator, a Hadamard-like condition can be imposed on it as a further 
regularity condition, in the spirit of Radzikowski \cite{Radzikowski1996}. 
Namely, one can require its wave front set to be minimal with respect to 
the subset relation, along with a positivity condition. It is seen that 
statement~(B) is central for these.}

\section*{Acknowledgments}

The authors would like to thank Antal Jakov\'ac, Gergely Fej\H{o}s, J\'anos Balog, J\'anos Pol\'onyi 
for the valuable discussions on the theory of Wilsonian regularization and 
renormalization. We thank to Karol Kozlowski for discussions on exactly 
solvable models.

This work was supported in part by the Hungarian Scientific 
Research fund (NKFIH K-138152 and K-142423).
Z.~Tarcsay was supported by the J\'anos Bolyai Research Scholarship of 
the Hungarian Academy of Sciences, and by the
\'UNKP--22-5-ELTE-1096 New National Excellence Program of the Ministry
for Innovation and Technology, as well as Project no.\  TKP2021-NVA-09, which 
has been implemented with the support provided by the Ministry of Culture and Innovation 
of Hungary from the National Research, Development and Innovation Fund, 
financed under the TKP2021-NVA funding scheme.

\appendix

\section{Some facts on distributions}
\label{secApp}

Throughout this Appendix, the notations of Section~\ref{secBF} are used. 
In particular, the base manifold $\M$ is a finite dimensional real affine space. 
(Without loss of generality, one may assume $\M:=\RR^{N}$.) 
Moreover, instead of generic coarse-graining operators, merely convolution 
operators by test functions are used, as a special case. 
Also, for simplifying the notations, without loss of generality, only scalar 
valued smooth functions, test functions and distributions are discussed here.

\begin{remark}[some complications of topological vector spaces]
Recall that for $n\in\N_{0}$, we use the notation $\EEn$ for the space of 
${\times}^{n}\M\rightarrow\RR$ smooth functions with their standard 
smooth function topology, and $\DDn$ for the compactly supported functions 
from these with their standard test function topology. 
The spaces $\EE$ and $\EEn$ are known to be nuclear Fr\'echet (NF) spaces 
(see \cite{Treves1970}~Theorem51.5 and its Corollary). 
The spaces $\DD$ and $\DDn$ are known to be countable strict inductive limit 
of NF spaces with closed adjacent images (LNF space, see \cite{Treves1970}~Ch13-6~ExampleII). 
It is customary to denote by 
$\otimes^{n}\EE$ and $\otimes^{n}\DD$ the $n$-fold algebraic tensor product 
of $\EE$ and $\DD$ with themselves, by $\otimes_{\pi}^{n}\EE$ and $\otimes_{\pi}^{n}\DD$ 
these spaces equipped with the so-called projective tensor product topology, 
moreover by $\hat{\otimes}_{\pi}^{n}\EE$ and $\hat{\otimes}_{\pi}^{n}\DD$ 
the topological completions of these. 
The Schwarz kernel theorem says that $(\hat{\otimes}_{\pi}^{n}\EE')'$ and $\hat{\otimes}_{\pi}^{n}\EE$ and $\EEn$ 
are naturally topologically isomorphic, moreover that 
$(\hat{\otimes}_{\pi}^{n}\EE)'$ and $\hat{\otimes}_{\pi}^{n}\EE'$ and $\EEn'$ are naturally topologically 
isomorphic (\cite{Treves1970}~Theorem51.6 and its Corollary). 
The distributional version of the Schwarz kernel theorem says that 
the spaces $\hat{\otimes}_{\pi}^{n}\DD'$ and $\DDn'$ 
are naturally topologically isomorphic (\cite{Treves1970}~Theorem51.7), 
moreover that there is a natural continuous linear bijection 
$(\hat{\otimes}_{\pi}^{n}\DD)'\rightarrow\DDn'$ 
(\cite{Horvath1966}~Chapter4.8~Proposition7). 
Care must be taken, however, that its inverse map is not continuous 
(\cite{Hirai2001}~Theorem2.4 and Remark2.1), 
i.e.\ the pertinent natural map is not a topological isomorphism. 
The corresponding transpose of the above statement says that 
the spaces $(\hat{\otimes}_{\pi}^{n}\DD')'$ and $\DDn$ 
are naturally topologically isomorphic, and that there is the 
natural continuous linear bijection $\DDn\rightarrow\hat{\otimes}_{\pi}^{n}\DD$, 
but its inverse map fails to be continuous. 
For this reason, one should distinguish in notation the spaces 
$\hat{\otimes}_{\pi}^{n}\DD$, $\DDn$ and 
$(\hat{\otimes}_{\pi}^{n}\DD)'$, $\DDn'$, respectively, due to their different topologies. 
That is, on the spaces $\DDn$ or $\DDn'$, there are multiple 
complete nuclear Hausdorff locally convex vector topologies which are 
comparable and inequal. 
On the $\EEn$ or $\EEn'$ type spaces, such 
complication is not present, due to their metrizability or dual metrizability, 
respectively. 
Also, these complications are absent if the above spaces are regarded rather 
as convergence vector spaces \cite{Beattie1996}.
\end{remark}

\begin{lemma}[a form of Lagrange lemma]
For all $\omega\in\DDn'$, the property $\forall g\in\DD:\,{\otimes}^{n}C_{g}\,\omega=0$ implies $\omega=0$.
(Therefore, such statement is also true when $\omega\in\EEn$.)
\label{lemLagrange}
\end{lemma}

\begin{proof}
Whenever $\omega\in\DDn'$ is arbitrary and $g_{i}\in\DD$ ($i\in\N_{0}$) is a Dirac delta 
approximating sequence, then the sequence ${\otimes}^{n}C_{g_{i}}\,\omega\in\EEn\subset\DDn'$ ($i\in\N_{0}$) 
converges to $\omega\in\DDn'$ in the weak-* topology. If 
$\omega\in\DDn'$ were such that $\forall g\in\DD:\,{\otimes}^{n}C_{g}\,\omega=0$ holds, 
then for a Dirac delta approximating sequence as above, the sequence 
${\otimes}^{n}C_{g_{i}}\,\omega\in\EEn\subset\DDn'$ ($i\in\N_{0}$) would be 
all zero, therefore its weak-* limit would be zero, being equal to $\omega$ by means of 
the above observation. Therefore, $\omega=0$ would follow.
\end{proof}

\begin{lemma}[the key lemma]
Let $\omega_{i}\in\DD_{m+n}'$ ($i\in\N_{0}$) be a sequence of distributions 
converging pointwise on the subspace $\DD_{m}\otimes\DDn$ of $\DD_{m+n}$. 
Then, it converges pointwise on the full $\DD_{m+n}$.
\label{lemotimes}
\end{lemma}

\begin{proof}
Let $\Psi\in\DD_{m+n}$, then there exists compact sets $\K\subset\times^{m}\M$ and $\ELL\subset\times^{n}\M$, 
such that $\Psi\in\DD_{m+n}(\K\times\ELL)\equiv\DD_{m}(\K)\hat{\otimes}_{\pi}\DDn(\ELL)$, 
with $\DD_{m+n}(\K\times\ELL)$ and $\DD_{m}(\K)$ and $\DDn(\ELL)$ being the 
corresponding nuclear Fr\'echet spaces of smooth functions with restricted support. 
Moreover, one has the identity
\begin{eqnarray}
 \Psi = \sum_{j\in\N_{0}} \lambda_{j}\, \varphi_{j} {\otimes} \psi_{j} \quad (\forall j\in\N_{0}:\; \lambda_{j}\in\RR,\; \varphi_{j}\in\DD_{m}(\K),\; \psi_{j}\in\DDn(\ELL))
\end{eqnarray}
where the sum is absolutely convergent in the $\DD_{m+n}(\K\times\ELL)$ topology, 
the sequence $\lambda_{j}\in\RR$ ($j\in\N_{0}$) is absolutely summable, 
and the sequence $\varphi_{j}\in\DD_{m}(\K)$ ($j\in\N_{0}$) as well as the sequence 
$\psi_{j}\in\DDn(\ELL)$ ($j\in\N_{0}$) are convergent to zero in the 
$\DD_{m}(\K)$ and $\DDn(\ELL)$ topology, respectively (\cite{Treves1970}~ChIII.45~Theorem45.1). 
Therefore, the pertinent convergences also hold in the spaces 
$\DD_{m+n}$ and $\DD_{m}$ and $\DDn$, respectively, due to the definition of 
their topologies. Using this, one infers
\begin{eqnarray}
 \forall i\in\N_{0}:\quad (\omega_{i}\,\vert\,\Psi) \,=\, \Big(\omega_{i} \,\Big\vert\, \sum_{j\in\N_{0}} \lambda_{j}\, \varphi_{j}{\otimes}\psi_{j} \Big) \,=\, \sum_{j\in\N_{0}} \lambda_{j}\, (\omega_{i} \,\vert\, \varphi_{j}{\otimes}\psi_{j} )
\label{eqLebesgue}
\end{eqnarray}
due to the continuity of the linear maps $\omega_{i}:\,\DD_{m+n}\rightarrow\RR$ ($i\in\N_{0}$). 
Moreover, due to the assumptions of the theorem, one has
\begin{eqnarray}
 \forall j\in\N_{0}:\quad \mathrm{the\; real\; valued\; sequence }\; i\mapsto(\omega_{i}\,\vert\,\varphi_{j}{\otimes}\psi_{j}) \;\mathrm{ is\; convergent}.
\label{eqLimit}
\end{eqnarray}
At the end of the proof we will show that the set of coefficients
\begin{eqnarray}
 \big\{\; (\omega_{i}\,\vert\,\varphi_{j}\otimes\psi_{j})\in\RR \;\,\big\vert\;\, i,j\in\N_{0} \;\big\} \;\subset\; \RR
\label{eqBounded}
\end{eqnarray}
is bounded. This fact implies that there exists a $C\in\RR^{+}$ such that 
$\forall\, i,j\in\N_{0}:\;\big\vert\,\lambda_{j}\;(\omega_{i}\,\vert\,\varphi_{j}\otimes\psi_{j})\,\big\vert\leq\vert\lambda_{j}\vert\,C$ holds, 
where the majorant sequence on the right hand side is absolutely summable due to our previous observations. 
Then, Lebesgue's theorem of dominated convergence for the exchange of limits and 
infinite sums on the two-index sequence $\lambda_{j}\;(\omega_{i}\,\vert\,\varphi_{j}\otimes\psi_{j})\in\RR$ ($i,j\in\N_{0}$) 
implies that the real valued sequence $i\mapsto\sum_{j\in\N_{0}} \lambda_{j}\; (\omega_{i} \,\vert\, \varphi_{j}\otimes\psi_{j})$ 
is convergent, the real valued sequence $j\mapsto\lim_{i\in\N_{0}} \,\lambda_{j}\; (\omega_{i} \,\vert\, \varphi_{j}\otimes\psi_{j})$ 
is absolutely summable, moreover 
${\lim_{i\in\N_{0}}\sum_{j\in\N_{0}} \lambda_{j}\; (\omega_{i} \,\vert\, \varphi_{j}\otimes\psi_{j} )} \;=\; {\sum_{j\in\N_{0}} \lim_{i\in\N_{0}} \;\lambda_{j}\; (\omega_{i} \,\vert\, \varphi_{j}\otimes\psi_{j} )}$ holds. 
This finding, in combination with Eq.(\ref{eqLebesgue}), yields that 
the real valued sequence $i\mapsto(\omega_{i}\,\vert\,\Psi)=\sum_{j\in\N_{0}} \lambda_{j}\; (\omega_{i} \,\vert\, \varphi_{j}\otimes\psi_{j} )$ 
is convergent, and that proves the theorem. In order to complete the proof, we show that the set of coefficients Eq.(\ref{eqBounded}) is indeed bounded.

According to the distributional Schwartz kernel theorem, 
$\DD_{m+n}'\equiv\mathcal{L}in(\DD_{m},\DDn')$ (\cite{Treves1970}~Theorem51.7). In this identification, 
by the assumptions of the theorem, the sequence of continuous linear maps 
${\omega_{i}:\,\DD_{m}\rightarrow\DDn'}$ ($i\in\N_{0}$) is convergent 
pointwise to zero, when the target space $\DDn'$ is equipped with the 
weak-* topology. Since $\DDn'$ is Montel, then the pertinent sequence of 
continuous linear maps is also convergent pointwise to zero, when the target space $\DDn'$ 
is equipped with its canonical strong dual topology. Therefore, the set 
of continuous linear maps ${\{\omega_{i}:\DD_{m}\rightarrow\DDn'\,\vert\,i\in\N_{0}\}}$ is 
pointwise bounded. Since the starting space 
$\DD_{m}$ is barrelled (\cite{Treves1970}~ChII.33~Corollary3), by means of Banach--Steinhaus theorem, this pointwise 
bounded set of continuous linear maps is equicontinuous (\cite{Treves1970}~ChII.33~Theorem33.1). 
Therefore, its image of the bounded set $\{\varphi_{j}\,\vert\,j\in\N_{0}\}\subset\DD_{m}$, 
being the set $\big\{(\omega_{i}\,\vert\,\varphi_{j}{\otimes}{\,\cdot})\in\DDn'\;\big\vert\;i,j\in\N_{0}\big\}\subset\DDn'$ 
is bounded (\cite{Rudin1991}~ChI.2~Theorem2.4). This argument can be repeated, 
namely, the elements of $\DDn'$ can be identified with $\DDn\rightarrow\RR$ 
continuous linear maps, and the set of continuous linear maps 
$\big\{(\omega_{i}\,\vert\,\varphi_{j}{\otimes}{\,\cdot}):\DDn\rightarrow\RR\;\big\vert\;i,j\in\N_{0}\big\}$ 
is pointwise bounded by means of our previous observation. 
Since $\DDn$ is barreled, by means of Banach--Steinhaus 
theorem, this pointwise bounded set of continuous linear maps 
is equicontinuous. Therefore, its image of the bounded set $\{\psi_{k}\,\vert\,k\in\N_{0}\}\subset\DDn$, being 
the set $\big\{(\omega_{i}\,\vert\,\varphi_{j}{\otimes}\psi_{k})\in\RR\;\big\vert\;i,j,k\in\N_{0}\big\}\subset\RR$ is 
bounded. Consequently, its subset Eq.(\ref{eqBounded}) is bounded, which completes the proof.
\end{proof}

It is well known that due to the distributional Banach--Steinhaus theorem, whenever a 
sequence of distributions $\omega_{i}\in\DDn'$ ($i\in\N_{0}$) is pointwise convergent over 
$\DDn$, then the pointwise limit mapping itself is a distribution. 
Lemma~\ref{lemotimes} implies that this can be generalized to $\otimes^{n}\DD$, as stated below.

\begin{lemma}[a Banach--Steinhaus-like theorem]
Let $\omega_{i}\in\DDn'$ ($i\in\N_{0}$) be a sequence of distributions 
which is pointwise convergent on the subspace $\otimes^{n}\DD$ of $\DDn$. 
Then, there exists a unique distribution $\Omega\in\DDn'$ such that 
$(\omega_{i}-\Omega)\in\DDn'$ ($i\in\N_{0}$) is pointwise convergent to zero on the full $\DDn$.
\label{lemAlgTensB}
\end{lemma}

\begin{proof}
We prove the theorem by induction. Clearly, the statement is true for $n=1$ due 
to the ordinary distributional Banach--Steinhaus theorem. Assume that the statement of the theorem 
holds for some $n\in\N_{0}$, and take a sequence of distributions 
$\omega_{i}\in\DD_{n+1}'$ ($i\in\N_{0}$) which is pointwise convergent on the 
subspace $\otimes^{n+1}\DD$ of $\DD_{n+1}$. Then, for all $\varphi\in\DD$ the 
sequence of distributions $(\omega_{i}\,\vert\,{\cdot\,}{\otimes}\varphi)\in\DDn'$ ($i\in\N_{0}$) 
is pointwise convergent on the subspace $\otimes^{n}\DD$ of $\DDn$. Therefore, 
by the induction assumption, there exists a unique distribution $\Omega_{\varphi}\in\DDn'$ such that 
$\big(\,(\omega_{i}\,\vert\,{\cdot\,}{\otimes}\varphi)-\Omega_{\varphi}\,\big)\in\DDn'$ ($i\in\N_{0}$) 
converges pointwise to zero on the full $\DDn$. Therefore, the sequence of distributions 
$\omega_{i}\in\DD_{n+1}'$ ($i\in\N_{0}$) is convergent pointwise on the subspace 
$\DDn\otimes\DD$ of $\DD_{n+1}$. By means of Lemma~\ref{lemotimes} it follows 
then that it converges pointwise over the full $\DD_{n+1}$. Applying the distributional 
Banach--Steinhaus theorem it follows that the statement of the theorem also holds 
for $n+1$, which completes the induction.
\end{proof}

\section*{References}

\bibliographystyle{JHEP}
\bibliography{wilsonian}

\end{document}